\documentclass[conference]{IEEEtran}
\usepackage{graphicx,indentfirst}
\usepackage{booktabs}
\usepackage{indentfirst}
\usepackage{amsmath}
\usepackage{cite}
\usepackage{multirow}
\usepackage{slashbox}
\usepackage{subfigure}
 \usepackage{mathrsfs}
\usepackage{amsmath}
\usepackage[english]{babel}
\usepackage{amsthm}
\usepackage{multirow}
\usepackage{algorithm, algorithmic}
\usepackage{paralist}
\usepackage{amsthm}
\usepackage{amssymb}
\usepackage{amsfonts}
\usepackage{color}
\newtheorem{proposition}{Proposition}
\newtheorem{theorem}{Theorem}
\newtheorem{lemma}{Lemma}
\newtheorem{assumption}{Assumption}

\newtheorem{definition}{Definition}

\IEEEoverridecommandlockouts

\begin{document}
\bibliographystyle{IEEE2}
\title{A Socially-Aware Incentive Mechanism for \\Mobile Crowdsensing Service Market}
\author{Jiangtian Nie$^{1,2}$, Zehui Xiong$^2$, Dusit Niyato$^2$, Ping Wang$^3$ and Jun Luo$^2$\\
$^1$Energy Research Institute @NTU, Interdisciplinary Graduate School, \\Nanyang Technological University, Singapore\\
$^2$School of Computer Science and Engineering, Nanyang Technological University, Singapore\\
$^3$Department of Electrical Engineering and Computer Science, York University, Canada }
\maketitle
\begin{abstract}
Mobile Crowdsensing has shown a great potential to address large-scale problems by allocating sensing tasks to pervasive Mobile Users (MUs). The MUs will participate in a Crowdsensing platform if they can receive satisfactory reward. In this paper, in order to effectively and efficiently recruit sufficient MUs, i.e., participants, we investigate an optimal reward mechanism of the monopoly Crowdsensing Service Provider (CSP). We model the rewarding and participating as a two-stage game, and analyze the MUs' participation level and the CSP's optimal reward mechanism using backward induction. At the same time, the reward is designed taking the underlying social network effects amid the mobile social network into account, for motivating the participants. Namely, one MU will obtain additional benefits from information contributed or shared by local neighbours in social networks. We derive the analytical expressions for the discriminatory reward as well as uniform reward with complete information, and approximations of reward incentive with incomplete information. Performance evaluation reveals that the network effects tremendously stimulate higher mobile participation level and greater revenue of the CSP. In addition, the discriminatory reward enables the CSP to extract greater surplus from this Crowdsensing service market.
\end{abstract}

\section{Introduction}\label{Sec:Introduction}
We are witnessing a fast proliferation of mobile users and devices in daily life. The ubiquitous mobile devices with various embedded functional sensors have remarkably promoted the information generation process. These advances stimulate the rapid development of mobile sensing technologies, and mobile Crowdsensing becomes one of the most attractive and popular paradigms. Classical examples include Amazon Mechanic Turk~\cite{AMT}, Waze~\cite{Waze}, Sensorly~\cite{Sensorly}, GreenGPS~\cite{GreenGPS}, etc. The Crowdsensing systems heavily rely on total user participation level and the individual contribution from each user. However, individuals are reluctant to participate and share their collected information due to the lack of sufficient motivation and incentive. To stimulate and recruit users with mobile devices to participate in Crowdsensing, the Crowdsensing Service Provider (CSP) usually provides a reward for the users as monetary incentive. For this purpose, a large number of prior works have been dedicated to designing incentive mechanisms~\cite{duan2012incentive, koutsopoulos2013optimal, kawajiri2014steered, luo2015crowdsourcing,peng2015pay, han2016taming, han2016truthful,zhang2016incentives,han2018quality}.

\begin{figure}[t]
\centering
\includegraphics[width=.50\textwidth]{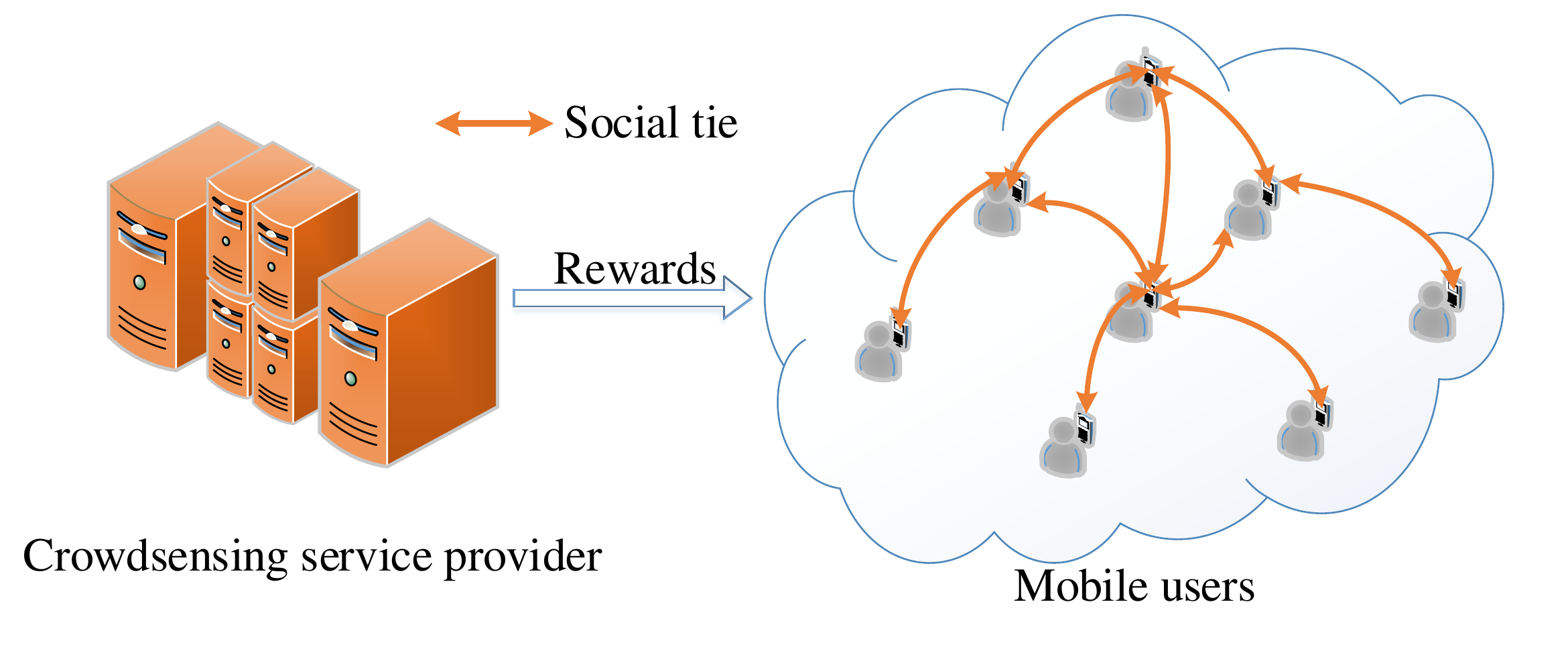}
\caption{A Mobile Crowdsensing Market with Social Network Effects.}\label{Fig:Model}
\end{figure}

In~\cite{duan2012incentive}, the authors proposed a reward-based collaboration mechanism, where the CSP announces a total reward to be shared among collaborators, and the task and reward are allocated if sufficient number of participants are willing to collaborate.
In~\cite{koutsopoulos2013optimal}, the authors presented a mechanism for participation level determination and reward allocation using optimal reverse auction, in which the CSP receives service queries and initiates an auction for user participation.
The authors developed a new framework called Steered Crowdsensing in~\cite{kawajiri2014steered}, which controls incentives by introducing gamifications with monetary reward to location-based services.
The authors in~\cite{luo2015crowdsourcing} applied Tullock contests to design incentive mechanisms, where the reward includes a fixed contest prize, and Tullock prize function depending on the winner's contribution.
In~\cite{peng2015pay}, the authors incorporated the consideration of data quality into the mechanism, and rewarded the participant depending on the quality of its collected data. The authors in~\cite{han2016taming,han2018quality} also considered the data quality and studied the tradoff between the recruiting cost and sensing robustness.
The comprehensive survey can be found in~\cite{zhang2016incentives}.

Unfortunately, all the papers addressed the incentive mechanism for Crowdsensing service without considering the interdependent behaviors of mobile users in social networks. Consequently, the complex user behaviors post a remarkable challenge to the operation of the Crowdsensing platform. Nevertheless, to the best of our knowledge, only one recent work~\cite{chen2016incentivizing} studied the incentive mechanism for Crowdsensing and exploited the network effects at the same time. However, the authors in~\cite{chen2016incentivizing} investigated the behaviors of mobile users under the global network effects\footnote{Global network effect refers to the phenomenon that a user will obtain higher value when its behavior aligns with any other users~\cite{easley2010networks}.}, which is not appropriate for the structure of an underlying social network. Social (local) network effects refer to the case where each user is only influenced directly by the decisions of other densely socially-connected users~\cite{xiong2017sponsor, xiong2018sponsor}. For example, in a mobile Crowdsensing platform for sharing road traffic information, a user (driver) can get a better route if more neighbourhood users (users in the same or nearby road) of this user join and contribute their traffic data~\cite{Waze, GreenGPS}. On the contrary, this user cannot obtain any benefits if the users in other distant roads join and share their traffic data. Therefore, this motivates us to explore the role of social network effects on the incentive mechanism of Crowdsensing service. The main contributions of this paper are:
\begin{itemize}
 \item We model the interaction between a CSP and mobile users as a two-stage game and analyze each stage systematically through backward induction;
 \item We exploit the social (local) network effects in the model, which, for the first time, utilizes the structural properties of the underlying social network, and fully characterizes the heterogeneity of the mobile participants;
 \item We propose the optimal incentive mechanism in terms of discriminatory reward and uniform reward with complete information, from which we obtain the analytical expression for optimal reward. Further, we derive the approximations of optimal reward with incomplete information;
 \item Performance evaluation demonstrates that the network effects strongly motivate the mobile users' participation, and thus improve the revenue of CSP.
\end{itemize}

The rest of this paper is organized as follows. Section~\ref{Sec:Model} describes the system model and the game formulation. In Section~\ref{Sec:Solution}, we analyze the mobile user participation level and optimal reward using backward induction under complete and incomplete scenarios. Section~\ref{Sec:Simulation} presents the performance evaluation. Section~\ref{Sec:Conclusion} summarizes the paper.

\section{System description and game formulation}\label{Sec:Model}
We model the interaction among Crowdsensing Service Provider (CSP) and the socially-aware participants, i.e., Mobile Users (MUs) as a simultaneous one-shot game, where the MUs' actions are to choose the participation level and the CSP provides the payment as a reward to incentivize the MUs. Consider a set of MUs ${\cal N} \buildrel \Delta \over = \{ 1, \ldots ,N\}$. Each MU $i \in {\cal N}$ determines its participation level, i.e., the effort level in participation (e.g.,  sensing data transmission frequency or sensing resolution), denoted by $x_i$ where ${x_i} \in (0,+ \infty)$.

Let $\mathbf{x} \buildrel \Delta \over = ({x_1}, \ldots ,{x_N})$ and $\mathbf{x}_{-i}$ denote the participation levels of all the MUs and all other MUs except MU $i$, respectively. The reward per effort unit provided to the MUs is given as: ${\bf{r}} = {[{{r}_{1}}, \ldots {{r}_{i}}, \ldots ,{{r}_{N}}]^\top}$. Then, the utility of MU $i$ is given by:
\begin{equation}\label{Eq:1}
{u_i}({x_i},{{\bf{x}}_{ - {{i}}}}) = {f_i}({x_i}) + {\Phi}({x_i},{{\bf{x}}_{ - {{i}}}}) + {r(x_i)} - {c}({x_i}).
\end{equation}
The first term ${f_i}(x)$ represents the private utility or internal effects that MU $i$ obtained from the participation, which can be formulated as ${f_i}(x_i) = {a_i}{x_i} - {b_i}{x_i}^2$, where $a_i > 0$ and $b_i > 0$ are the coefficients that capture the intrinsic value of the participation to different MUs with heterogeneity~\cite{candogan2012optimal, xiong2017sponsor, gong2015network}. For example, a healthcare Crowdsensing platform enables an MU to obtain a better understanding of its health conditions by keeping track of the diet, exercise and heart rate~\cite{swan2012health}. Additionally, the linear-quadratic function captures the decreasing marginal returns from participation.

The second term, $\Phi({x_i},\mathbf{{x}}_{-i})$ is denoted as the external benefits gained from the network effects, which is the key component from Eq.~(\ref{Eq:1}). In social networks, an MU can enjoy an additional benefit from information contributed or shared by the others~\cite{easley2010networks}. The existing work in crowdsensing explored the network effects of global nature, where the additional benefits due to new coming MUs are the same for all the existing MUs~\cite{chen2016incentivizing}. However, due to the structural properties of the underlying social network, it is more appropriate to consider the network effects locally in Crowdsensing service, i.e., the social network effects. Then, we introduce the adjacency matrix $\mathcal{G} = {[{g_{ij}}]_{i, j \in \mathcal{N}}}$. The elements in matrix $g_{ij}$ indicates the influence of MU $j$ on MU $i$, which can be unidirectional or bidirectional. In this paper, we consider the case where $g_{ij}= g_{ji}$, i.e., the social tie is reciprocal. Specifically, we adopt ${\sum _{j \in \cal N}}{g_{ij}}{x_i}{x_j}$ to represent the additional benefits obtained from the network effects, similar to that in~\cite{candogan2012optimal, xiong2017sponsor, xiong2018sponsor}.

The third term, ${r(x_i)}$, is the reward from CSP to the MU $i$, which is equal to ${r_i}{x_i}$ with the participation level based reward mechanism. The last term ${c}({x_i})$ denotes the cost associated to the participation level of the MU, e.g., energy consumption and network bandwidth consumed. Similar to~\cite{chen2016incentivizing}, we assume that the cost is equal to $c{x_i}$, where $c$ is the MU's unit cost. Then the utility of MU $i$ is expressed by:
\begin{equation*}\label{Eq:2}
{u_i}({x_i},{{\bf{x}}_{ - {{i}}}}, {\bf{r}}) =  {a_i}{x_i} - {b_i}{x_i}^2 +\sum\nolimits_{i = 1}^N {{g_{ij}}{x_i}{x_j}}   + {r_i}{x_i} - {c}{x_i}.
\end{equation*}

We now assume that the monopoly CSP operates and maintains the platform with a fixed cost, which is ignored for the simplicity of the analysis later. Then, the formulation of revenue for the CSP is given by the profit from total aggregated contribution of all MUs minus the total reward paid to MUs:
\begin{equation}\label{Eq:3}
\Pi   = \mu\sum\nolimits_{i=1}^N {( s{x_i} - t{x_i}^2)} - \sum\nolimits_{i=1}^N {r_i}{x_i}.
\end{equation}
Similar to~\cite{xiong2018sponsor}, we also use the linear-quadratic function to transform the MUs' participation level to the monetary revenue of CSP, which features the law of diminishing return: an MU's contribution increases with the MU's effort level but the marginal return decreases. If all the MUs do not contribute any effort, the utility received by the CSP is $s0-t0^2=0$. $\mu$ is an adjustable parameter representing the equivalent monetary worth of MUs' participation level, and $s, t> 0$ are coefficients capturing the concavity of the function.

We model the interaction between the CSP and the MUs using a two-stage single-leader-multiple-follower Stackelberg game as follows.
\begin{definition}{Two-stage reward-participation game:}
\begin{itemize}
 \item Stage I (Reward): The CSP determines the reward, aiming at the highest revenue:
   \[{{\bf{r}}^*} = \arg \max_{\bf{r}} \left\{\mu\sum\nolimits_{i=1}^N{( s{x_i} - t{x_i}^2)} - \sum\nolimits_{i=1}^N {r_i}{x_i}\right\};\]
 \item Stage II (Participation): Each MU $i \in {\cal N}$ chooses the participation level $x_i$, given the observed reward $\bf r$ and the participation levels of other MUs ${{\bf{x}}_{ - {{i}}}}$, with the goal to maximize its individual utility:     \[x_i^* = \arg \mathop {\max_{x_i}}  {u_i}({x_i},{{\bf{x}}_{ - {{i}}}}, \bf{r}).\]
\end{itemize}
\end{definition}
We solve this two-stage game by finding a subgame perfect equilibrium for the cases of discriminatory reward mechanism and uniform reward mechanism for all MUs. Furthermore, we differentiate between the complete information scenario where the CSP has knowledge about all $\{a_i\}_{i=1}^N$ and $\{b_i\}_{i=1}^N$, and the scenario in which only their expectations $\mathbb E[a]$ and $\mathbb E[b]$ are known by the CSP.

\section{Game equilibrium analysis}\label{Sec:Solution}
\subsection{Stage II: MUs' participation equilibrium}
Based on the definition of the Nash equilibrium, each MU chooses its participation level that is the best response. By setting the first-order derivative $\frac{{\partial {u_i}({x_i},{{\bf{x}}_{ - i}})}}{{\partial {x_i}}}$ to $0$, we obtain the best response of MU $i$ as follows:
\begin{equation}\label{Eq:4}
x_i^* =\max \left\{ 0, \frac{{{r_i} - c + {a_i}}}{{2{b_i}}} + \sum\nolimits_{j=1}^{N}\frac{{ {{g_{ij}}} }}{{2{b_i}}}{x_j} \right\}, \forall i.
\end{equation}
Each MU's best response includes two parts: $\frac{{{r_i} - c + {a_i}}}{{2{b_i}}} $ is independent from the strategies of the other MUs, and $\sum\nolimits_{j=1}^{N}\frac{{ {{g_{ij}}} }}{{2{b_i}}}{x_j}$ is dependent on the other MUs' participation levels due to underlying social network effects. Although the participation level strategy of each MU is obtained as in Eq.~(\ref{Eq:4}), the Nash equilibrium cannot be ensured to be unique or even exist since each MU may unboundedly increase its participation level if the other MUs' participation levels are large enough. Then, we present an sufficient assumption, under which there exists a unique Nash equilibrium as described in Theorem~1. Note that an MU has the limitation on participation level due to the battery capacity of a mobile device, and thus the Assumption~1 is reasonable.
\begin{assumption}
$\sum\nolimits_{j=1}^{N} {\frac{{{g_{ij}}}}{{{2b_i}}}}  < 1, \forall i$.
\end{assumption}
\begin{theorem}
Under Assumption~1, the existence and uniqueness of MU participation equilibrium, i.e., the Nash equilibrium of Stage~II in this Stackelberg game can be guaranteed.
\end{theorem}
\begin{proof}
\par
\textbf{Existence of MU participation equilibrium:} We denote $\bf x^*$ as the strategy profile in this MU participation sub-game, and $x^\dag_i$ as the largest participation level in $\bf x^*$. Then, we have
\begin{eqnarray}
x^\dag_i &=& \left(\frac{{{r_i} - c + {a_i}}}{{2{b_i}}} + \sum\nolimits_{j=1}^{N}\frac{{ {{g_{ij}}} }}{{2{b_i}}}{x_j} \right)^+ \nonumber \\ &\le& \frac{{{r_i} - c + {a_i}}}{{2{b_i}}} + \sum\nolimits_{j=1}^{N} {{x^\dag_i}\frac{{{g_{ij}}}}{{2{b_i}}}}\nonumber \\ &\le& \frac{{\left| {{{r_i} - c + {a_i}}} \right|}}{{2{b_i}}} + \sum\nolimits_{j=1}^{N} {{x^\dag_i}\frac{{\left| {{g_{ij}}} \right|}}{{2{b_i}}}}.
\end{eqnarray}
Thus, under Assumption 1, we have
\begin{equation}
x^\dag_i  \le \frac{{\left|{{r_i} - c + {a_i}} \right|}}{{2{b_i} - \sum\nolimits_{j=1}^{N} {\left| {{g_{ij}} } \right|} }} =  \widehat x.
\end{equation}
As a result, the strategy space $[0, \widehat x]$ is convex and compact and the utility function ${u_i}({x_i},{\mathbf{x}}_{-i})$ is continuous in $x_i$ and ${\bf x}_{-i}$. We also have the second-order derivative of MU's objective function
\begin{equation}
\frac{{{\partial ^2}{u_i}}}{{{\partial ^2}{x_i}}}=-2b_i<0.
\end{equation}
Thus, the MU participation sub-game is a concave game which admits the Nash equilibrium.
\par
\textbf{Uniqueness of MU participation equilibrium:} Firstly, we have $ - \frac{{{\partial ^2}{u_i}}}{{{\partial}{x_i}^2}} =  - (-2{b_i} + {g_{ii}}) = 2{b_i}$. Then, based on Assumption~1, we have $- \frac{{{\partial ^2}{u_i}}}{{{\partial}{x_i}^2}} > \sum\nolimits_{j=1}^N {{g_{ij}}} = \sum\nolimits_{j=1}^N {\left| {{g_{ij}}} \right|} \nonumber= \sum\nolimits_{j=1}^N{\left| { - \frac{{{\partial ^2}{u_i}}}{{\partial {x_i}{x_j}}}} \right|}$, which satisfies the dominance solvability condition, i.e., Moulin's Theorem~\cite{moulin1984dominance}. As a result, the uniqueness of MU participation equilibrium is guaranteed under Assumption~1. The proof is then completed.
\end{proof}
Then, we propose the best response dynamic algorithm to obtain the Nash equilibrium with respect to MUs' participation level, as shown in Algorithm 1. The algorithm iteratively updates MUs' strategies
based on their best response functions in Eq.~(\ref{Eq:4}), and converges to the Nash equilibrium of MU participation sub-game.
\begin{algorithm}\footnotesize
 \caption{Simultaneous best-response updating for finding Nash equilibrium of MU participation sub-game}
 \begin{algorithmic}[1]
 \STATE \textbf{Input:} \\
 Precision threshold $\varepsilon$, $x_i^{[0]} \leftarrow 0 $, $x_i^{[1]} \leftarrow 1 + \varepsilon$, $k\leftarrow 1$;
  \WHILE {$\left\|x_i^{[k]} - x_i^{[k-1]}\right\|_1>\varepsilon$}
  \FORALL {$i \in \cal N$}
   \STATE  $x_i^{[k+1]} = \left( \frac{{{r_i} - c + {a_i}}}{{2{b_i}}} + \sum\nolimits_{j=1}^{N} {{x_j^{[k]}}\frac{{{g_{ij}} }}{{2{b_i}}}} \right)^+$;
  \ENDFOR
  \STATE $k\leftarrow k + 1 $;
  \ENDWHILE
  \STATE \textbf{Return} ${\bf{x}}_i^{[k]}$;
 \end{algorithmic}\label{algorithm}
\end{algorithm}

\begin{proposition}
Algorithm 1 achieves the Nash equilibrium of MU participation sub-game.
\end{proposition}
\begin{proof}
Let ${\kappa _{ij}} = \frac{{{g_{ij}}}}{{2{b_i}}}$ and $\Delta x_i^{[k]}$ = $x_i^{[k]} - x_i^*$, for all $i \in \cal N$. From step 4 of Algorithm 1, we have
\begin{eqnarray}
\left| {\Delta x_i^{[k]}} \right| &\le& \left| {\sum\nolimits_{j=1}^{N} {{\kappa _{ij}}} \Delta x_j^{[k]}} \right| \nonumber \\ &\le& \sum\nolimits_{j=1}^{N} {{\kappa _{ij}}} \left| {\Delta x_j^{[k]}} \right|.
\end{eqnarray}
We also have
\begin{equation}
{\left\| {\Delta {{\bf{x}}^{[k]}}} \right\|_\infty } = \mathop {\max }\limits_i \left| {\Delta x_i^{[k]}} \right|.
\end{equation}
Accordingly, we conclude that
\begin{eqnarray}
{\left\| {\Delta {{\bf{x}}^{[k]}}} \right\|_\infty } &\le& \mathop {\max }\limits_i \left( {\left| {\sum\nolimits_{j=1}^{N} {{\kappa _{ij}}} \Delta x_j^{[k]}} \right|} \right) \nonumber \\ &\le& \mathop {\max }\limits_i \sum\nolimits_{j=1}^{N} {{\kappa _{ij}}} \mathop {\max }\limits_i \left( {\left| {\Delta x_j^{[k]}} \right|} \right) \nonumber \\&=& \mathop {\max }\limits_i \sum\nolimits_{j=1}^{N} {{\kappa _{ij}}} {\left\| {\Delta {{\bf{x}}^{[k]}}} \right\|_\infty }.
\end{eqnarray}
Under Assumption 1, we have $\sum\nolimits_{j=1}^{N} {\kappa _{ij}}< 1, \forall i$. Thus, ${\left\| {\Delta {x^{[k]}}} \right\|_\infty } \le {\left\| {\Delta {x^{[k - 1]}}} \right\|_\infty }$, and the algorithm leads to a contraction mapping of ${\left\| {\Delta {x^{[k - 1]}}} \right\|_\infty }$. The proof is completed.
\end{proof}
For the ease of presentation, we have the following definitions, ${\bf B}:=diag(2b_1, 2b_2, \ldots, 2b_N)$, ${\bf a}:= [a_i]_{N \times 1}$, ${\bf 1}:= [1]_{N \times 1}$, ${\bf G}:=[g_{ij}]_{N \times N}$, ${\bf r}:= [r_i]_{N \times 1}$ and ${\bf{I}}$ is the $N \times N$ identity matrix. For the rest of the paper, similar to~\cite{xiong2017sponsor, zhou2017peer}, we consider the ideal situation where all the MUs have the positive participation levels at the Stackelberg equilibrium, i.e., a special case of Eq.~(\ref{Eq:4}). Then, with Lemma~1, we can rewrite Eq.~(\ref{Eq:4}) in a matrix form as follows:
\begin{equation}
{\bf{x}} = {\bf K}\left( {{\bf{a}} + {\bf{r}} - c{\bf{1}}} \right),
\end{equation}
where ${\bf K} = {\left( {{\bf{B}} - {\bf{G}}} \right)^{ - 1}}$.
\begin{lemma}
${{\bf{B}} - {\bf{G}}}$ is positive definite matrix, which is invertible.
\end{lemma}
\begin{proof}
If Assumption~1 holds, we have
\begin{eqnarray}
{\left( {{\bf{B}} - {\bf{G}}} \right)_{ii}} &=& 2{b_i} - {g_{ii}} = 2{b_i}\nonumber\\ &>& \sum\nolimits_{j=1, j\ne i}^{N} {{g_{ij}}} = \sum\nolimits_{j=1}^{N} {{g_{ij}}} \nonumber \\ &=&- \sum\nolimits_{j=1}^{N} {{{\left( {{\bf{B}} - {\bf{G}}} \right)}_{ij}}} \nonumber \\ &=&  - \sum\nolimits_{j=1}^{N} {\left| {{{\left( {{\bf{B}} - {\bf{G}}} \right)}_{ij}}} \right|}.
\end{eqnarray}
Accordingly, ${{\bf{B}} - {\bf{G}}}$ is strictly diagonal dominant and all the diagonal elements, i.e., $2b_i$ are larger than $0$. Based on Gershgorin circle theorem~\cite{weisstein2003gershgorin}, every eigenvalue $\lambda$ of ${{\bf{B}} - {\bf{G}}}$ satisfies
\begin{equation}
\left| {{{\left( {{\bf{B}} - {\bf{G}}} \right)}_{ii}} - \lambda } \right| < \sum\nolimits_{j=1}^{N} {\left| {{{\left( {{\bf{B}} - {\bf{G}}} \right)}_{ij}}} \right|}.
\end{equation}
Accordingly, we know $\lambda > 0$, and thus ${{\bf{B}} - {\bf{G}}}$ is a positive definite matrix, from which its invertibility follows. The proof is completed.
\end{proof}
\subsection{Stage I: Optimal reward mechanism}
In this stage, the monopoly CSP determines the reward to be paid to the MUs, the objective of which is to maximize its overall revenue. We first focus on the discriminatory reward mechanism, uniform reward mechanism with the full information scenario. Furthermore, we investigate the single reward mechanism in the incomplete information scenario.
\subsubsection{Discriminatory reward mechanism with complete information}
Under reward discrimination, the CSP is able to provide different levels of reward for MUs to maximize its revenue. The revenue maximization problem can be formulated as follows:
\begin{equation}
\begin{aligned}
& \underset{\bf r}{\text{maximize}}
& & {\Pi} = \mu\sum\nolimits_{i=1}^{N}{( s{x_i} - t{x_i}^2)} - \sum\nolimits_{i=1}^{N} {r_i}{x_i}\\
& & & \quad = \mu( s{\bf{1}}^\top{\bf{x}}- {\bf{x}}^\top t{\bf{x}}) - {\bf{r}}^\top{\bf{x}}\\
& \text{subject to}
& & {\bf{x}} = {\bf{K}}\left( {\bf{a}} + {\bf{r}} - c{\bf{1}}\right).\\
\end{aligned}\label{Eq:5}
\end{equation}
By plugging $\bf x$ into the objective function in Eq.~(\ref{Eq:5}), we have $\Pi = \mu \left(s{{\bf{1}}^ \top }{\bf{K}}\left( {{\bf{a}} + {\bf{r}} - c{\bf{1}}} \right) - t{\left( {{\bf{a}} + {\bf{r}} - c{\bf{1}}} \right)^ \top }{{\bf{K}}^2}\left( {{\bf{a}} + {\bf{r}} - c{\bf{1}}} \right)\right) - {{\bf{r}}^ \top }{\bf{K}}\left( {{\bf{a}} + {\bf{r}} - c{\bf{1}}} \right)$. Setting the first-order condition of objective function in Eq.~(\ref{Eq:5}) with respect to $\bf r$ to zero, i.e., $\frac{{\partial \Pi }}{{\partial {\bf{r}}}} = 0$, we obtain $\mu \left(s{\bf{K1}} - 2t{{\bf{K}}^2}\left( {{\bf{a}} + {\bf{r}} - c{\bf{1}}} \right)\right) - {\bf{K}}\left( {{\bf{a}} + {\bf{r}} - c{\bf{1}}} \right) - {\bf{Kr}} = 0$. Then, we have $\mu \left(s{\bf{K1}} - 2t{{\bf{K}}^2}\left({{\bf{a}} - c{\bf{1}}} \right)\right) - {\bf{K}}\left( {{\bf{a}} - c{\bf{1}}} \right) = \left( {2{\bf{K}} + 2\mu t{{\bf{K}}^2}} \right){\bf{r}}$. Finally, we obtain the optimal value ${\bf r}^*$, which is represented as follows:
\begin{equation}\label{Eq:6}
{\bf r}^* = {\left(2{\bf{I}}+ 2\mu t{\bf{K}}\right)^{ - 1}}\left\{ \mu \left[s{\bf{1}} - 2t{\bf{K}}\left( {{\bf{a}} - c{\bf{1}}} \right)\right] - \left( {{\bf{a}} - c{\bf{1}}} \right)\right\}.
\end{equation}
\subsubsection{Uniform reward mechanism with complete information}
In this case, the CSP can only choose a single uniform reward to be paid to all the MUs, i.e., $r_i = r$, for all $i$. Then, the optimization problem is given by
\begin{equation}
\begin{aligned}
& \underset{r}{\text{maximize}}
& & {\Pi} = \mu\sum\nolimits_{i=1}^{N}{( s{x_i} - t{x_i}^2)} - {r}\sum\nolimits_{i=1}^{N}{x_i}\\
& & & \quad = \mu( s{\bf{1}}^\top{\bf{x}}- {\bf{x}}^\top t{\bf{x}}) - {{r}}{\bf{1}}^\top{\bf{x}}\\
& \text{subject to}
& & {\bf{x}} = {\bf{K}}\left[{\bf{a}} + (r - c){\bf{1}}\right].\\
\end{aligned}\label{Eq:7}
\end{equation}
Similarly, we eliminate $\bf x$ from objective function in Eq.~(\ref{Eq:7}), then we obtain $\Pi = \mu \left(s{{\bf{1}}^ \top }{\bf{K}}\left[ {{\bf{a}} + (r - c){\bf{1}}} \right] - t{\left[ {{\bf{a}} + (r - c){\bf{1}}} \right]^ \top }{{\bf{K}}^2}\left[ {{\bf{a}} + (r - c){\bf{1}}} \right]\right)$ $- r{{\bf{1}}^ \top }{\bf{K}}\left[ {{\bf{a}} + (r - c){\bf{1}}} \right]$. Then, we evaluate its first-order optimality condition with respect to the reward $r$, which yields that $\frac{{\partial \Pi }}{{\partial r}} = \mu \{ (s{{\bf{1}}^ \top }{\bf{K1}} - 2t{\left[ {{\bf{a}} + (r - c){\bf{1}}} \right]^ \top }{{\bf{K}}^2}{\bf{1}}\}  - {{\bf{1}}^ \top }{\bf{K}}\left[ {{\bf{a}} + (r - c){\bf{1}}} \right] - r{{\bf{1}}^ \top }{\bf{K1}} = 0$. As a result, we obtain the optimal value of the uniform reward, which is represented by
\begin{multline}\label{Eq:8}
{r^*} = {\left(2\mu t{{\bf{1}}^ \top }{{\bf{K}}^2}{\bf{1}} + 2{{\bf{1}}^ \top }{\bf{K}}\right)^{ - 1}}\Big\{ \mu \big[s{{\bf{1}}^ \top }{\bf{K1}} - 2t{({\bf{a}} - c{\bf{1}})^ \top }
\\ \times {{\bf{K}}^2}{\bf{1}}\big]  - {{\bf{1}}^ \top }{\bf{K}}\left({\bf{a}} - c{\bf{1}}\right)\Big\}.
\end{multline}
\subsubsection{Uniform reward mechanism with incomplete information}
We next explore the uniform reward mechanism for the incomplete information scenario, similar to that in~\cite{zhou2017peer}. Under the additional assumption as follows, we obtain the lower bound of single reward, i.e., single reward bound as in Theorem~2.
\begin{assumption}
$c \ge {\bar a} + \mu s$, where $\bar a = \sum\nolimits_{i = 1}^N {{a_i}} /N.$
\end{assumption}
\begin{theorem}
In the scenario of incomplete information, i.e., only the expectations of $\{a_i\}_{i=1}^n$ and $\{b_i\}_{i=1}^n$ are known and denoted by ${\mathbb E}[a]$ and ${\mathbb E}[b]$, the optimal uniform reward is bounded as shown in Eq.~(\ref{Eq:10}).
\end{theorem}
\begin{proof}
Since we have the optimal value of ${{\bf{r}}^*} = {\left( {2{\bf{I}} + 2\mu t{\bf{K}}} \right)^{ - 1}}\left\{ {\mu \left[ {s{\bf{1}} - 2t{\bf{K}}\left( {{\bf{a}} - c{\bf{1}}} \right)} \right] - \left( {{\bf{a}} - c{\bf{1}}} \right)} \right\}$, we can easily obtain the average value of total offered reward, i.e., $\overline r  =\sum\nolimits_i {{r_i}}/N$. Specifically, we derive the expression of $\overline r$, as shown in Eq.~(\ref{Eq:long}). After taking the expectation with respect to the random variables $\{a_i\}_{i=1}^N$ and $\{b_i\}_{i=1}^N$, we obtain the analytical expression for the optimal single reward bound with incomplete information in Eq.~(\ref{Eq:10}), where $E_b(\Xi)$ is to take the expected value for random variable $b_i$ in the component of $\Xi$.
\begin{figure*}[ht]
\begin{equation}\label{Eq:long}
\overline r  = \frac{1}{N}{{\bf{1}}^ \top }{\left( {2{\bf{I}} + 2\mu t{\bf{K}}} \right)^{ - 1}}\left\{ {\mu \left[ {s{\bf{1}} - 2t{\bf{K}}\left( {{\bf{a}} - c{\bf{1}}} \right)} \right] - \left( {{\bf{a}} - c{\bf{1}}} \right)} \right\} = \frac{{ - {{\bf{1}}^ \top }\left( {{\bf{a}} - c{\bf{1}}} \right)}}{{2N}} + \frac{{{{\bf{1}}^ \top }}}{N}{\left( {2{\bf{I}} + 2\mu t{\bf{K}}} \right)^{ - 1}}\left[ {\mu s{\bf{1}} - \mu t{\bf{K}}\left( {{\bf{a}} - c{\bf{1}}} \right)} \right].
\end{equation}
\begin{eqnarray}\label{Eq:9}
{r_u} &=& \frac{{c - \overline a }}{2} + {E_b}\left( {\frac{1}{N}{{\bf{1}}^ \top }{{\left[ {2\left( {{\bf{B}} - {\bf{G}}} \right) + 2\mu t{\bf{I}}} \right]}^{ - 1}}\left[ {\mu s\left( {{\bf{B}} - {\bf{G}}} \right) - \mu t\left( {\overline a  - c} \right){\bf{I}}} \right]{\bf{1}}} \right) \nonumber \\
&=& \frac{{c - \bar a}}{2} + {E_b}\left( {\frac{1}{N}{{\bf{1}}^ \top }{{\left[ {2\left( {{\bf{B}} - {\bf{G}}} \right) + 2\mu t{\bf{I}}} \right]}^{ - 1}}\left[ {\mu s\left( {{\bf{B}} - {\bf{G}}} \right) + \mu {s^2}t{\bf{I}} - \mu {s^2}t{\bf{I}} - \mu t\left( {\bar a - c} \right){\bf{I}}} \right]{\bf{1}}} \right) \nonumber\\
&=& \frac{{c - \bar a}}{2} + \frac{{\mu s}}{{2N}}{{\bf{1}}^ \top }{\bf{1}} + {E_b}\left( {\frac{1}{N}{{\bf{1}}^ \top }{{\left[ {2\left( {{\bf{B}} - {\bf{G}}} \right) + 2\mu t{\bf{I}}} \right]}^{ - 1}}\left[ { - {\mu ^2}st - \mu t\left( {\bar a - c} \right)} \right]{\bf{1}}} \right) \nonumber\\
&=& \frac{{c - \bar a}}{2} + \frac{{\mu s}}{2} - \frac{{\mu t}}{N}\left( {\mu s + \bar a - c} \right){E_b}\left( {{{\bf{1}}^ \top }{{\left[ {2\left( {{\bf{B}} - {\bf{G}}} \right) + 2\mu t{\bf{I}}} \right]}^{ - 1}}{\bf{1}}} \right)
\end{eqnarray}
\begin{equation}\label{Eq:10}
\Rightarrow{r_u} \ge \frac{{c - {\mathbb E}\left[ a \right]}}{2} + \frac{{\mu s}}{2} - \frac{{\mu t}}{N}\left( {\mu s +{\mathbb E}\left[ a \right] - c} \right){\bf{1}}^{\top}{\left[ {2{\mathbb E}{_b}\left[ {\bf{B}} \right] - 2{\bf{G}} + 2\mu t{\bf{I}}} \right]^{ - 1}}{\bf{1}}
\end{equation}
\hrulefill
\end{figure*}
If $c \ge {\bar a} + \mu s$ holds, we denote $${\bf D} = 2{\bf B}_1  - 2{\bf G} +2\mu t {\bf I} >0$$ and $${\bf E} = 2{\bf B}_2 - 2{\bf G} +2\mu t {\bf I} >0,$$ and we have
\begin{multline*}
\beta \underbrace {\left[ {\begin{array}{*{20}{c}}
{{{\bf{1}}^ \top }{\bf{D1}}}&{{{\bf{1}}^ \top }}\\
{\bf{1}}&{\bf{D}}
\end{array}} \right]}_{ \ge 0} + (1 - \beta )\underbrace {\left[ {\begin{array}{*{20}{c}}
{{{\bf{1}}^ \top }{\bf{E1}}}&{{{\bf{1}}^ \top }}\\
{\bf{1}}&{\bf{E}}
\end{array}} \right]}_{ \ge 0} =\\
 {\left[ {\begin{array}{*{20}{c}}
{\beta {{\bf{1}}^ \top }{{\bf{D}}^{ - 1}}{\bf{1}} + (1 - \beta ){{\bf{1}}^ \top }{{\bf{E}}^{ - 1}}{\bf{1}}}&{{{\bf{1}}^ \top }}\\
{\bf{1}}&{\beta {\bf{D}} + (1 - \beta ){\bf{E}}}
\end{array}} \right]}  \ge 0,
\end{multline*}
where $\beta \in (0, 1)$. Using Schur decomposition, we obtain that
\begin{equation}
\beta {{\bf{1}}^ \top }{{\bf{D}}^{ - 1}}{\bf{1}} + (1 - \beta ){{\bf{1}}^ \top }{{\bf{E}}^{ - 1}}{\bf{1}} \ge {{\bf{1}}^ \top }\left[ {\beta {\bf{D}} + (1 - \beta ){\bf{E}}} \right]{\bf{1}}.
\end{equation}
Then, we know that $f({\bf{X}}) = {{{\bf{1}}^ \top }{\bf{X}}^{-1} \bf{1}}$ is convex for all positive definite matrix $\bf{X}$. Based on Jensen's inequality in the multivariate case on random variable $\bf X$,
\begin{equation}
{{\mathbb E}_Y}\left[g({\bf{X}})\right] \ge g\left[{{\mathbb E}_Y}\left({\bf{X}}\right)\right],
\end{equation}
we obtain the lower bound of optimal reward with incomplete information, i.e., single reward bound as shown in Eq.~(\ref{Eq:10}). The proof is completed.
\end{proof}

\section{Performance evaluation}\label{Sec:Simulation}
\begin{figure}[t]
\centering
\includegraphics[width=.45\textwidth]{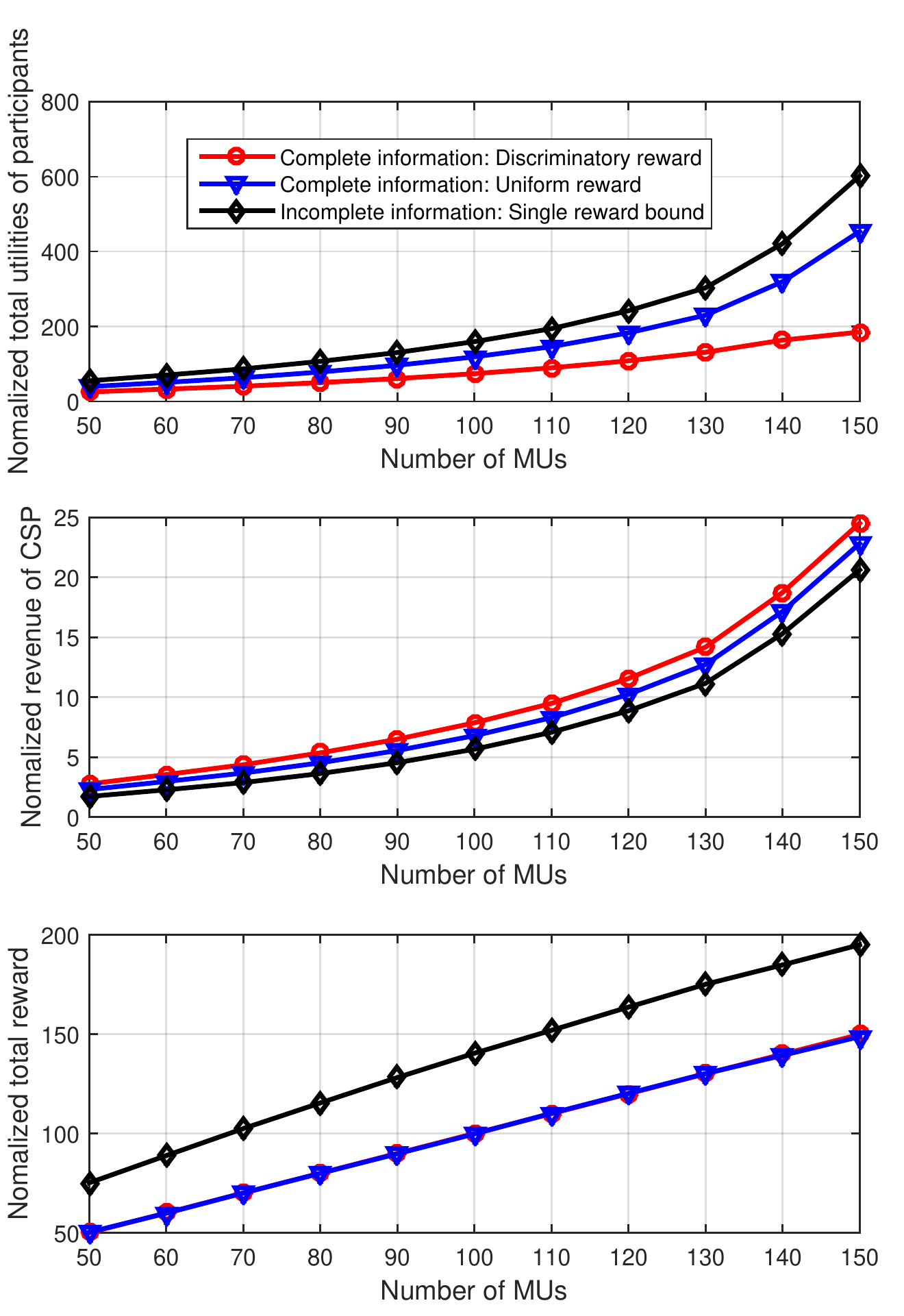}
\caption{The impact of total number of MUs on the Crowdsensing service provider and mobile participants.}\label{Fig:number}
\end{figure}
\begin{figure}[t]
\centering
\includegraphics[width=.45\textwidth]{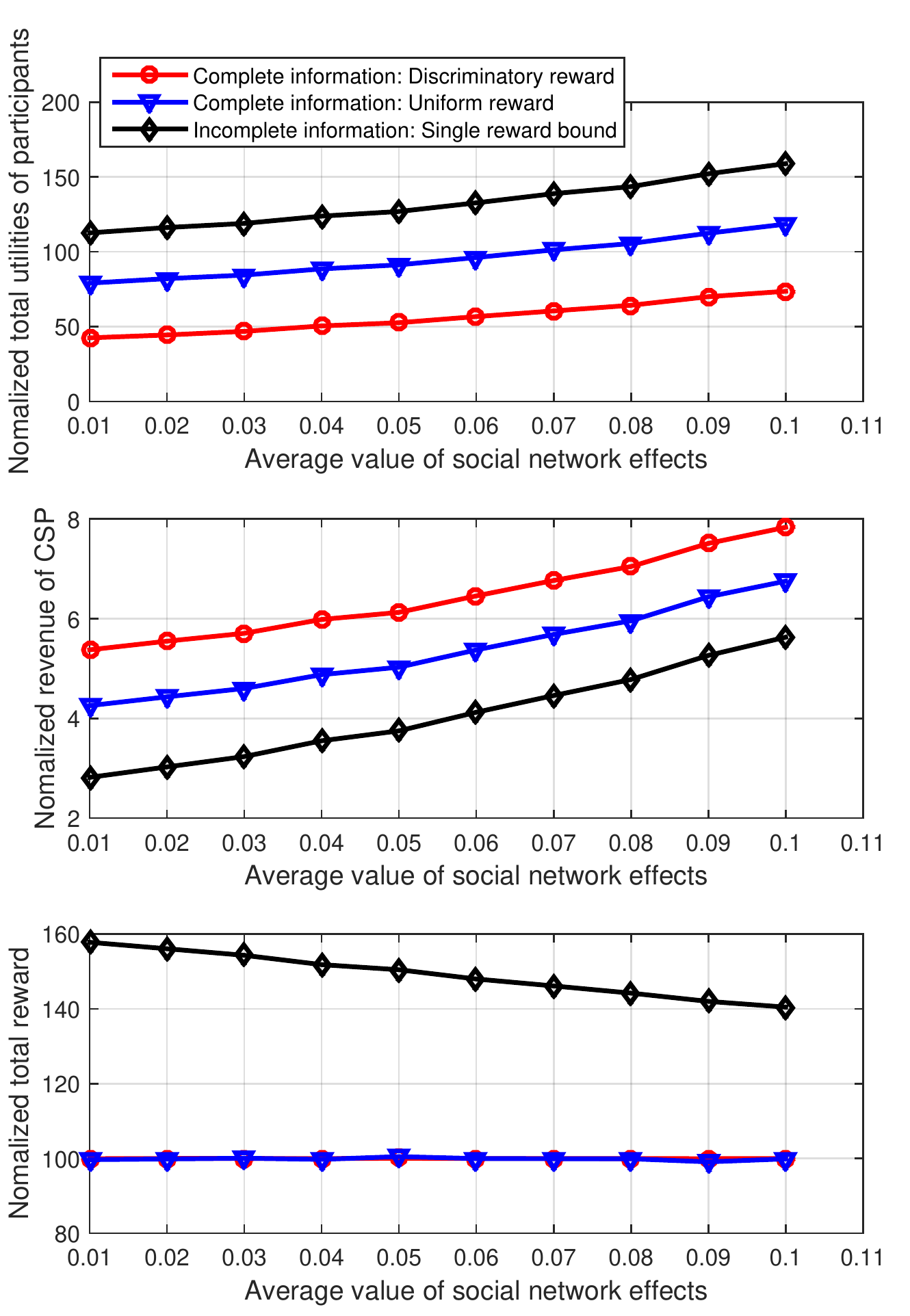}
\caption{The impact of average value of social network effects on the Crowdsensing service provider and mobile participants.}\label{Fig:socialtie}
\end{figure}
In this section, we evaluate the performance of the proposed incentive mechanism of CSP, and investigate the impacts of different parameters of mobile social networks on the performance. We consider a group of $N$ MUs, i.e., mobile participants, in a social network and set the parameters as follows. We assume the intrinsic parameters of MUs, $a_i$ and $b_i$ follow the normal distribution $\mathcal{N}(\mu_a, 2.5)$ and $\mathcal{N}(\mu_b,2.5)$. In addition, the social tie $g_{ij}$ between any two participants $i$ and $j$ follows a normal distribution $\mathcal{N}(\mu_g, 1)$. The default parameters are set as: $c=15+\alpha$, $\alpha=1$, $\mu=0.01$, $s=20$, $t=0.05$, $\mu_a=\mu_b=15$, $\mu_g=0.1$ and $N=100$. Note that some of these parameters are varied according to the evaluation scenarios.

As expected and verified in Fig.~\ref{Fig:number} and Fig.~\ref{Fig:socialtie}, the discriminatory reward under complete information scenario yields the largest revenue for the CSP, compared with other mechanisms. Intuitively, this is because the CSP can adjust the reward according to individual MU's effort and contribution, which is proven by Fig.~\ref{Fig:index}. Meanwhile, setting the lower bound reward under incomplete information scenario does not sacrifice too much profit, compared with the uniform reward mechanism, especially when the number of MUs is higher and the social network effects is stronger.

We first evaluate the impact of the total number of MUs on the proposed three incentive mechanisms, as illustrated in Fig.~\ref{Fig:number}. As the number of MUs increases, the total utilities of participants and the revenue of CSP also increase under all these mechanisms. The marginal increase of the total utilities of participants and the revenue of CSP are also greater, as the total number of MUs is higher. This is because when the total number of MUs increases, the social neighbourhood MUs of each MU also increases. Owing to the underlying social network effects, the MUs are motivated by their social neighbours to have higher participation levels, and the revenue of CSP is improved accordingly. In addition, with the increase of total number of mobile participants, the total offered reward increases since the CSP needs to remunerate more MUs to participate, in order to attain a satisfying revenue gain. In particular, the discriminatory and uniform reward mechanisms enable the CSP to reduce the reward, i.e., the cost, and therefore achieve a higher revenue gain in turn. Fig.~\ref{Fig:socialtie} depicts the impact of average value of social network effects on two entities of this network, i.e., the CSP and MUs. We observe that as the social network effects become stronger, the total utilities of participants and the revenue of CSP also increase. Since when the social ties is stronger, the additional benefits obtained from social network effects are greater. In other words, the socially-aware MUs are motivated by each other and have higher participation levels consequently. When the participation levels are high enough, the CSP is able to offer less reward to save money. In turn, the total utilities of participants and the revenue of CSP are improved.

\begin{figure}[t]
\centering
\includegraphics[width=.45\textwidth]{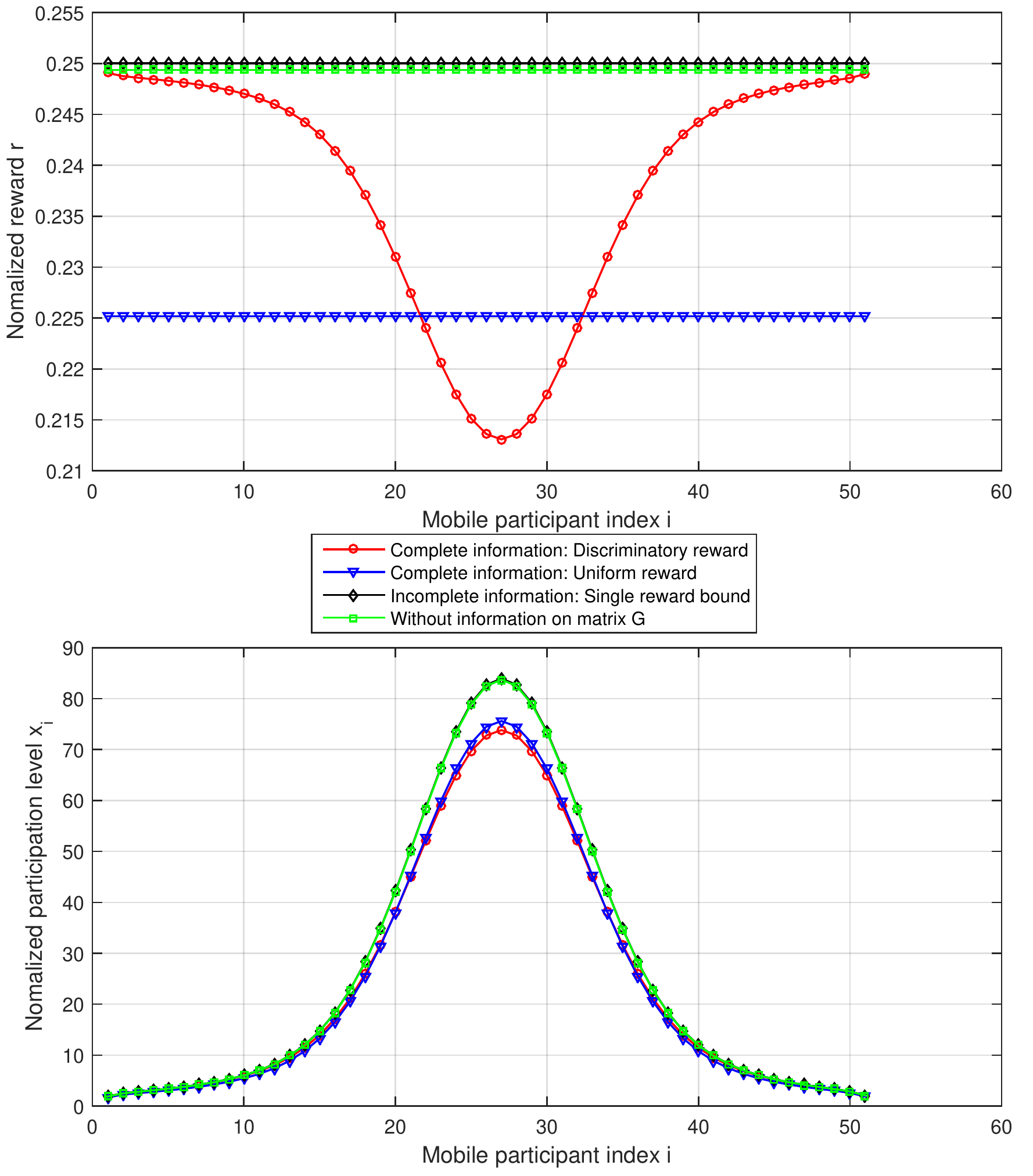}
\caption{A case illustration of distribution of normalized reward and participation level.}\label{Fig:index}
\end{figure}
Then, to explore the impacts of social network effects on each specific participant, we investigate the optimal reward and resulting MUs' participation level, and we adopt the following parameters. $N=51$, $b=0.1$, $a=15$, $c=15 + \alpha$, $\alpha=1$, $\mu=0.01$, $s=50$, $t=0.05$. The adjacency matrix $G$ is generated based on the rules, shown as follows:
\begin{equation}\label{Eq:simulation}
\begin{cases}
   {g_{i,i + 1}} = 0.2\times\left(0.5 - {{\left(0.5 - \frac{{i - 1}}{N}\right)}^2}\right), &\mbox{$i \in [1,N - 1]$};\\
   {g_{i + 1,i}} = 0.2\times\left(0.5 - {{\left(0.5 - \frac{{i - 1}}{N}\right)}^2}\right), &\mbox{$i \in [1,N - 1]$};\\
   {g_{i,j}} = 0,&\mbox{otherwise}.
\end{cases}
\end{equation}

From Eq.~(\ref{Eq:simulation}), only participants who are adjacent in participant indexes (neighbours) can affect each other. From Fig.~\ref{Fig:index}, we notice that the CSP offers each participant the same reward when it knows nothing about matrix $G$. Given the reward incentive from the CSP, the MUs have different participation level equilibrium shown in Fig.~\ref{Fig:index}, where we observe that the participation levels of the MUs are socially related to each other. In particular, the $27$th MU is the most influenced one in this network because it has the highest participation level given the same reward incentive. On the contrary, the $1$st and the $51$st MUs are the most influencing ones. Therefore, with the knowledge on matrix $G$, the CSP is likely to offer more reward to the $1$st and the $51$st MUs and less to the $27$th MU, under the discriminatory reward incentive mechanism. This is because the CSP wants to have the highest participation level from the participant with the lowest cost and thus have a higher revenue gain. However, under the other two reward mechanisms, the CSP can only offer the same reward to the participants. In this case, the CSP usually offers higher reward and promotes the participants to attain higher participation level, but the incurred extra cost is also very high, which decreases the revenue of CSP correspondingly.

\section{Conclusion}\label{Sec:Conclusion}
In this work, we have developed a two-stage game theoretic model, and analyzed each stage using backward induction. The Crowdsensing Service Provider (CSP) sets the incentive mechanism in the first stage, and the Mobile Users (MUs) decide their participation level in response to the observed reward in the second stage. Taking the social (local) network effects among the MUs into account, we have proposed the optimal incentive mechanism, i.e., discriminatory reward and uniform reward under complete information scenario, where we have obtained the closed-form expression for optimal reward. Moreover, we have derived the approximations for reward incentive under incomplete information scenario. Performance evaluations have verified that the network effects significantly improve the participation level and the revenue of CSP. Further, it has been validated that the discriminatory reward mechanism helps the CSP to achieve greater revenue gain. The future work will extend the monopoly CSP setting to a multiple CSP competition. Another direction is to take spatial/temporal concerns into reward mechanisms for the Crowdsensing such as~\cite{xiong2017network,han2016truthful}.

\section*{Acknowledgement}
This work was supported in part by National Research Foundation (NRF) Singapore, project NRF-ENIC-SERTD-SMES-NTUJTCI3C-2016, SMES, and AcRF Tier 2 Grant MOE2016-T2-2-022.

\bibliography{bibfile}

\begin{thebibliography}{10}

\bibitem{AMT}
``Amazon mechanic turk,'' https://www.mturk.com/mturk.

\bibitem{Waze}
``Waze,'' https://www.waze.com/.

\bibitem{Sensorly}
``Sensorly,'' http://www.sensorly.com/.

\bibitem{GreenGPS}
``Greengps,'' http://green-way.cs.illinois.edu/GreenGPS.html.

\bibitem{duan2012incentive}
L.~Duan, T.~Kubo, K.~Sugiyama, J.~Huang, T.~Hasegawa  and J.~Walrand,
\newblock ``Incentive mechanisms for smartphone collaboration in data
  acquisition and distributed computing,''
\newblock in {\em Proc. IEEE INFOCOM}, 2012.

\bibitem{koutsopoulos2013optimal}
I.~Koutsopoulos,
\newblock ``Optimal incentive-driven design of participatory sensing systems,''
\newblock in {\em Proc. IEEE INFOCOM}, 2013.

\bibitem{kawajiri2014steered}
R.~Kawajiri, M.~Shimosaka  and H.~Kashima,
\newblock ``Steered crowdsensing: incentive design towards quality-oriented
  place-centric crowdsensing,''
\newblock in {\em Proc. ACM Ubicomp}, 2014.

\bibitem{luo2015crowdsourcing}
T.~Luo, S.~S. Kanhere, H.-P. Tan, F.~Wu  and H.~Wu,
\newblock ``Crowdsourcing with tullock contests: A new perspective,''
\newblock in {\em Proc. IEEE INFOCOM}, 2015.

\bibitem{peng2015pay}
D.~Peng, F.~Wu  and G.~Chen,
\newblock ``Pay as how well you do: A quality based incentive mechanism for
  crowdsensing,''
\newblock in {\em Proc. ACM Mobihoc}, 2015.

\bibitem{han2016taming}
K.~Han, C.~Zhang  and J.~Luo,
\newblock ``Taming the uncertainty: Budget limited robust crowdsensing through
  online learning,''
\newblock {\em IEEE/ACM Transactions on Networking (TON)}, vol. 24, no. 3, pp.
  1462--1475, 2016.

\bibitem{han2016truthful}
K.~Han, C.~Zhang, J.~Luo, M.~Hu  and B.~Veeravalli,
\newblock ``Truthful scheduling mechanisms for powering mobile crowdsensing,''
\newblock {\em IEEE Transactions on Computers}, vol. 65, no. 1, pp. 294--307,
  2016.

\bibitem{zhang2016incentives}
X.~Zhang, Z.~Yang, W.~Sun, Y.~Liu, S.~Tang, K.~Xing  and X.~Mao,
\newblock ``Incentives for mobile crowd sensing: A survey,''
\newblock {\em IEEE Communications Surveys \& Tutorials}, vol. 18, no. 1, pp.
  54--67, 2016.

\bibitem{han2018quality}
K.~Han, H.~Huang  and J.~Luo,
\newblock ``Quality-aware pricing for mobile crowdsensing,''
\newblock {\em IEEE/ACM Transactions on Networking}, , no. 99, pp. 1--14, 2018.

\bibitem{chen2016incentivizing}
Y.~Chen, B.~Li  and Q.~Zhang,
\newblock ``Incentivizing crowdsourcing systems with network effects,''
\newblock in {\em Proc. IEEE INFOCOM}, 2016.

\bibitem{easley2010networks}
D.~Easley and J.~Kleinberg,
\newblock {\em Networks, crowds, and markets: Reasoning about a highly
  connected world},
\newblock Cambridge University Press, 2010.

\bibitem{xiong2017sponsor}
Z.~Xiong, S.~Feng, D.~Niyato, P.~Wang  and Y.~Zhang,
\newblock ``Economic analysis of network effects on sponsored content: a
  hierarchical game theoretic approach,''
\newblock in {\em Proc. IEEE GLOBECOM}, Singapore, 2017.

\bibitem{xiong2018sponsor}
Z.~Xiong, S.~Feng, D.~Niyato, P.~Wang  and Y.~Zhang,
\newblock ``Competition and cooperation analysis for data sponsored market: A
  network effects model,''
\newblock in {\em Proceedings of IEEE WCNC}, Barcelona, Spain, April 2018.

\bibitem{candogan2012optimal}
O.~Candogan, K.~Bimpikis  and A.~Ozdaglar,
\newblock ``Optimal pricing in networks with externalities,''
\newblock {\em Operations Research}, vol. 60, no. 4, pp. 883--905, 2012.

\bibitem{gong2015network}
X.~Gong, L.~Duan, X.~Chen  and J.~Zhang,
\newblock ``When social network effect meets congestion effect in wireless
  networks: Data usage equilibrium and optimal pricing,''
\newblock {\em IEEE JSAC}, vol. 35, no. 2, pp. 449--462, 2017.

\bibitem{swan2012health}
M.~Swan,
\newblock ``Health 2050: The realization of personalized medicine through
  crowdsourcing, the quantified self, and the participatory biocitizen,''
\newblock {\em Journal of personalized medicine}, vol. 2, pp. 93--118, 2012.

\bibitem{moulin1984dominance}
H.~Moulin,
\newblock ``Dominance solvability and cournot stability,''
\newblock {\em Mathematical Social Sciences}, vol. 7, pp. 83--102, 1984.

\bibitem{zhou2017peer}
D.~P. Zhou, M.~A. Dahleh  and C.~J. Tomlin,
\newblock ``How peer effects influence energy consumption,''
\newblock {\em arXiv preprint arXiv:1703.00980}, 2017.

\bibitem{weisstein2003gershgorin}
E.~W. Weisstein,
\newblock ``Gershgorin circle theorem,''
\newblock 2003.

\bibitem{xiong2017network}
Z.~Xiong, S.~Feng, D.~Niyato, P.~Wang  and Z.~Han,
\newblock ``Network effect-based sequential dynamic pricing for mobile social
  data market,''
\newblock in {\em Proc. IEEE GLOBECOM}, Singapore, 2017.

\end{thebibliography}

\end{document}